\documentclass[11pt]{article}
\usepackage{amssymb,amsmath,amsthm,amscd,latexsym}
\usepackage{mathrsfs}
\usepackage{mathrsfs}
\usepackage{amsfonts}
\usepackage{amsmath}
\usepackage{array}
\usepackage{cite}
 \usepackage{xcolor}
\usepackage{amssymb}
\usepackage{amscd}
\usepackage{makecell}
 \usepackage{color}
 \usepackage{bm}
\usepackage{extpfeil}
\usepackage{graphicx}
\usepackage{hyperref}

\renewcommand{\paragraph}{\roman{paragraph}}
\input cyracc.def

\newtheorem{thm}{\bfseries  Theorem}[section]
\newtheorem{lem}[thm]{\bfseries   Lemma}

\newtheorem{Def}[thm]{\bfseries  Definition}
\newtheorem{ex}[thm]{\bfseries   Example}
\newtheorem{rem}[thm]{\bfseries   Remark}

\begin{document}
%\begin{CJK*}{GBK}{song}\CJKtilde
\title{\bf New Constructions of Mutually Orthogonal Complementary Sets and Z-Complementary Code Sets Based on Extended Boolean Functions}
\author{Hongyang Xiao\thanks{ Hongyang Xiao, College of Mathematics, Nanjing University of Aeronautics and Astronautics, Nanjing, Jiangsu, 211106, China, {\tt xhycxyf@163.com}},
Xiwang Cao\thanks{Corresponding author. Xiwang Cao, College of Mathematics, Nanjing University of Aeronautics and Astronautics, Nanjing, Jiangsu, 211106, China; Key Laboratory of Mathematical Modeling and High Performance Computing of Air Vehicles (NUAA), MIIT, Nanjing, 211106, China, {\tt xwcao@nuaa.edu.cn}}
}
\date{}
\maketitle
%\end{CJK*}
%\begin{CJK}{GBK}{song}

\begin{abstract}
Mutually orthogonal complementary sets (MOCSs) and  Z-complementary code sets (ZCCSs) have many applications in practical scenarios such as synthetic aperture imaging systems and multi-carrier code division multiple access (MC-CDMA) systems.  With the aid of extended Boolean functions (EBFs), in this paper, we first propose a direct construction of MOCSs with flexible lengths, and then propose a new construction of ZCCSs. The proposed MOCSs cover many existing lengths and have non-power-of-two lengths when $q=2$. Our presented second construction can generate optimal ZCCSs meeting the set size upper bound. Note that the proposed two constructions are direct without the aid of any special sequence, which is suitable for rapid hardware generation.

\end{abstract}



\section{Introduction}
\textcolor{black}{The concept of Golay complementary pair (GCP) was initiated by Golay in 1961\cite{MG}.} The aperiodic auto-correlation function (AACF) of a GCP diminishes to zero for all time shifts except at zero.  In 1972, Tseng and Liu generalized the concept of GCP to Golay complementary sets (GCSs) and MOCSs \cite{CC1}. \textcolor{black}{An $(N,L)$-GCS is a set of $N$ $(\geq 2)$ sequences of length $L$ with the property that their AACF is zero for any non-zero time shifts and an $(M,N,L)$-MOCS is a collection of $M$ GCSs, in which every GCS has  $N$ sequences of length $L$ such that any two distinct GCSs are orthogonal.} In 1988, Suehiro and Hatori proposed the concept of $(N,N,L)$-complete complementary codes (CCCs) whose set size achieves the theoretical upper bound of MOCSs (i.e., $M\leq N$)\cite{SN}. Due to the ideal correlation properties,
MOCSs have been applied in many practical scenarios such as synthetic aperture imaging systems\cite{TY}, OFDM-CDMA systems \cite{ZZ} and MC-CDMA systems \cite{TS,AJ,LZ}. 

In recent years, the construction of MOCSs has attracted extensive attention in sequence design community. \textcolor{black}{Generalized Boolean functions (GBFs), usually are utilized to construct MOCSs. This is initiated by the pioneer work of Davis and Jedwab in \cite{JAD} which proposed a direct construction of $2^h$-ary $(h>0)$ GCPs of length $2^m$ $(m>0)$. Paterson extended the idea of \cite{JAD} to construct $q$-ary (for even $q$) GCPs \cite{PK}. Further constructions of GCPs and GCSs based on GBFs have been proposed in \cite{YL2,CY2}. 
In \cite{AR}, Tathinakumar and Chaturvedi proposed a direct construction of $q$-ary CCCs of length $2^m$ by extending Paterson's idea in \cite{PK}. Wu \textit{et al.} \cite{SW} designed MOCSs with non-power-of-two lengths.  Later, a number of direct constructions of $q$-ary MOCSs with non-power-of-two lengths are presented in \cite{LT,KM}. Sarkar \textit{et al.}  in \cite{SP5} proposed $(p^{n+1},p^{n+1},p^m)$-CCCs via $q$-ary functions ($\mathbb{Z}_{p}^{m}\rightarrow \mathbb{Z}_{q}$), where  $p$  is a prime number and $q$ is a positive multiple of $p$. But these CCCs only have  prime power lengths.
In \cite{PS2}, Sarkar \textit{et al.}  designed  CCCs of length $p_1^{m_1}p_2^{m_2}\cdots p_k^{m_k}$ (where each $p_i$ is a prime and $m_i$ is a  positive integer)  using multivariable functions (MVFs)\cite{PS2}. This direct construction can generate $q$-ary CCCs of all possible lengths. However, in the case of $q=2,$ only binary CCC of length of form $2^m$ $(m\in \mathbb{Z})$ has been constructed \cite{PS2}. Apart from these direct constructions of MOCSs, there are some other indirect methods to construct MOCSs such as interleaving, concatenation, paraunitary (PU) matrices, Kronecker product, extended correlation, etc.  \cite{XC,YJ,ZG,KL,SD1,SD2}. However, the generated MOCSs  may not be friendly for hardware generation due to their large space and time requirements. So how to  construct  MOCSs with flexible lengths is still an open problem.}

 \textcolor{black}{Since the set size is constrained by the number of sub-carrier in MOCSs, which prevents the communication system from supporting a large number of users, Fan \emph{et al.} proposed the concept of ZCCSs in \cite{FP}. The reason why ZCCSs have large set sizes is that there is a zero correlation zone (ZCZ) in the aperiodic cross-correlation and auto-correlation. For an $( M, N,L, Z)$-ZCCS, it holds that $M\leq N\lfloor L/Z\rfloor$ and it is optimal if the upper bound is achieved, \textcolor{black}{where $M,$ $N,$ $L,$ $Z$ refer  to the set size, number of sub-carrier, length and ZCZ width, respectively.} Especially, an $( M, N,L, Z)$-ZCCS is called an MOCS if $Z = L$. 
In the literature, ZCCSs are constructed by using direct and indirect methods. In \cite{WS}, Wu \textit{et al.} proposed ZCCSs with length $2^m$ ($m>0$) based on GBFs, then they expanded the parameters of ZCCSs in 2021 \cite{SW2}.  Several  GBFs based constructions of ZCCSs are presented in the literature \cite{SP1,SP2,SP3,GG1}. Additionally, Tian \textit{et al.} constructed ZCCSs by using PU matrices in \cite{LT1}. Yu \textit{et al.} applied Kronecker product to obtain ZCCSs in \cite{TY1}. Das \textit{et al.} presented a class of ZCCSs by using Butson-type Hadamard (BH) matrices and optimal Z-paraunitary (ZPU) matrices \cite{DS}. Adhikary  and Majhi in \cite{AA} employed Hadamard product to construct ZCCSs with new parameters. Further constructions of ZCCSs  have been proposed in \cite{YL,GG,SP4,RA}. In most previous designs, however, the optimal ZCCSs based on  direct methods  have limited lengths and the optimal ZCCSs based on indirect methods have limited hardware generations. Recently, Shen \textit{et al.} introduced the concept of EBF  and obtained optimal ZCCSs of length $q^m$ \cite{SB}, where $q\geq 2$ is a positive integer. This work helps us to construct more optimal ZCCSs.}

\textcolor{black}{Motivated by the existing works on MOCSs and ZCCSs,  in this paper, we construct $(q^{d'},q^{v+d},\gamma )$-MOCSs with flexible lengths and   $ (q^{v+d},q^{d},q^{m},q^{m-v}) $-ZCCSs by using EBFs, where $\gamma=a_mq^{m-1}+\sum^{v-1}_{k=1}a_{k}q^{m-v+k-1}+q^{u}$, $0<d'<d<m $, $v<m$, $a_k\in\mathbb{Z}_{q}$, $a_m\in\mathbb{Z}^*_{q}$ and $q\geq 2$ is a  positive integer. According to the arbitrariness of $q$, the proposed MOCSs cover the result in \cite{PS2} and have non-power-of-two lengths when $q=2$. In addition, the resulting MOCSs and ZCCSs can be obtained directly from EBFs without using  tedious sequence operations. Note that the proposed ZCCSs are  optimal with respect to the theoretical upper bound.}

The remainder of this paper is outlined as follows. In Section 2, we give the notations and definitions that will be used throughout this paper. In Section 3, we show  a construction of MOCS with flexible lengths. In Section 4, we present  an construction optimal ZCCS. In Section 5, we  make a comparison of the existing literature with this paper. Finally, we conclude this paper in Section 6.
\section{Preliminaries}
\subsection{Notation}
\begin{itemize}
	\item $\mathbb{Z}_{q}=\{0,1,\cdots,q-1\}$ is the ring of integers modulo $q$, where $q\geq 2$ is a positive integer throughout this paper, unless we specifically point out;
	\item $\mathbb{Z}^{*}_{q}=\mathbb{Z}_{q}\backslash \{0\}$;
	\item $\mathbb{N}_{m}=\{1,2,\cdots,m\}$ is the set with $m$ elements;
	\item $\xi=e^{2\pi\sqrt{-1}/q}$ is a primitive $q$-th root of unity;
	\item $\lfloor x\rfloor$ denotes the largest integer lower than or equal to $x$;
	\item $\lceil x\rceil$ denotes the smallest integer bigger than or equal to $x$;
	\item Bold small letter $\mathbf{a}$ denotes a sequence of length $L$, i.e., $\mathbf{a}=(a_{0}, a_{1},\cdots,a_{L-1});$
	\item $(\cdot)^{*}$ denotes the conjugate of  $(\cdot).$
\end{itemize}

\subsection{Correlation functions and complementary sequence sets}

Assume $\mathbf{a}=(a_{0},a_{1},\cdots,a_{L-1})$ and $\mathbf{b}=(b_{0},b_{1},\cdots,b_{L-1})$ are $\mathbb{Z}_{q}$-valued sequences of length $L$, where $a_{i}$ and $b_{i}$ are in the ring $\mathbf{\mathbb{Z}}_{q}$. The aperiodic cross-correlation function (ACCF) $R_{\mathbf{a},\mathbf{b}}(\tau)$ between $\mathbf{a}$ and $\mathbf{b}$ at a time shift $\tau$ is defined as
$$R_{\mathbf{a},\mathbf{b}}(\tau)
=\left\{\begin{array}{ll}
	\sum_{i=0}^{L-1-\tau}\xi^{a_{i}-b_{i+\tau}}, &0\leq \tau \leq L-1, \\
	\sum_{i=0}^{L-1+\tau} \xi^{a_{i-\tau}-b_{i}}, &-L+1 \leq \tau<0.
\end{array}\right. $$
If $\mathbf{a}=\mathbf{b}$, then $R_{\mathbf{a},\mathbf{b}}(\tau)$ is called the aperiodic auto-correlation function (AACF), denoted as $R_{\mathbf{a}}(\tau)$. In addition, by the definition of AACF, we get $R_{\mathbf{b},\mathbf{a}}(-\tau)=R^{*}_{\mathbf{a},\mathbf{b}}(\tau).$
\begin{Def}
	A set of $N$ length-$L$ sequences $\{\mathbf{a}_{0},\mathbf{a}_{1},\cdots,\mathbf{a}_{N-1}\}$ is called a GCS of order $N$ if for all $0<|\tau|\leq L-1,$
	$$\sum^{N-1}_{i=0}R_{\mathbf{a}_{i}}(\tau)=0.$$
\end{Def}

\begin{Def}
	A set of $M$ sequence sets $\mathcal{S}=\{\mathcal{S}_{0},\mathcal{S}_{1},\cdots,\mathcal{S}_{M-1}\}$ is called an $(M,N,L)$-MOCS if for any  $0\leq i\neq j\leq M-1$ and $0\leq|\tau|\leq L-1,$
	$$R_{\mathcal{S}_{i}, \mathcal{S}_{j}}(\tau)=\sum^{N-1}_{k=0}R_{\mathbf{a}^{i}_{k},\mathbf{a}^{j}_{k}}(\tau)=0,$$
	where each $\mathcal{S}_{t}=\{\mathbf{a}^{t}_{0},\mathbf{a}^{t}_{1},\cdots,\mathbf{a}^{t}_{N-1}\}$ is a GCS of $N$ length-$L$ sequences.
\end{Def}

\begin{Def}
	A set of $M$ sequence sets $\mathcal{S}=\{\mathcal{S}_{0},\mathcal{S}_{1},\cdots,\mathcal{S}_{M-1}\}$ is called an $(M,N,L,Z)$-ZCCS if
	$$\begin{aligned}
		R_{\mathcal{S}_{i}, \mathcal{S}_{j}}(\tau)
		&=\sum_{k=0}^{N-1} R_{\mathbf{a}^i_{k},\mathbf{a}^{j}_{k}}(\tau)
		=\left\{\begin{array}{ll}NL, & \tau=0, \;i=j, \\
			0, & 0<|\tau|<Z, \; i=j, \\
			0, & |\tau|<Z, \;i\neq j,
		\end{array}\right.\end{aligned}$$
\end{Def}
\noindent where $Z$ denotes the ZCZ width and each $\mathcal{S}_{t}=\{\mathbf{a}^{t}_{0},\mathbf{a}^{t}_{1},\cdots,\mathbf{a}^{t}_{N-1}\}$ consists of $N$ length-$L$ sequences for any $0\leq t\leq M-1$. In addition, if $Z=L$, then the $(M,N,L,Z)$-ZCCS is called an $(M,N,L)$-MOCS.

\textcolor{black}{The following results give two bounds on the parameters of MOCSs and ZCCSs, respectively.}
\begin{lem}\textrm{\cite{SN}}
	For an $(M,N,L)$-MOCS, the upper bound on set size satisfies the inequality $M\leq N$. When $M=N$, it is called a CCC.
\end{lem}

\begin{lem}\cite{FL} \label{15}
	For any $(M,N,L,Z)$-ZCCS, it holds that $$M\leq N\left \lfloor\frac{L}{Z}\right\rfloor.$$
Note that	a $ZCCS$ is $optimal$ if the above upper bound is achieved,i.e., $M= N\left \lfloor\frac{L}{Z}\right\rfloor.$
\end{lem}

\subsection{Extended Boolean functions (EBFs)}

An EBF $f$ in $m$ variables $x_{1},x_{2},\cdots,x_{m}$ is a mapping from $\mathbb{Z}_{q}^{m}$ to $\mathbb{Z}_{q}$ where $x_{i}\in \mathbb{Z}_{q}$ for $i\in 1,2,\cdots,m$. Given $f(x)$, we define \textcolor{black}{$$\mathbf{f}=(f_{0},f_{1},\cdots,f_{q^{m}-1}),$$}where $f_{i}=f(i_{1},i_{2},\cdots,i_{m})$ and $(i_{1},i_{2},\cdots,i_{m})$ is the $q$-ary representation of the integer $i=\sum^{m}_{k=1}i_{k}q^{k-1}
$.  For example, for $f=x_{1}x_{2}+x_{1}+2$ with $m=2$ and $q=3$, we have the sequence \textcolor{black}{$\mathbf{f}=(2,0,1,2,1,0,2,2,2).$} In addition, we also consider the sequences of length $L\neq q^m.$  Hence we define the corresponding truncated sequence \textcolor{black}{$\mathbf{f}^{(L)}$} of the EBF $f$ by removing the last $q^m-L$ elements of the sequence \textcolor{black}{$\mathbf{f}$}. That is \textcolor{black}{$\mathbf{f}^{(L)}=(f_{0},f_{1},\cdots,f_{L-1})$} is a sequence of length $L$ with $f_{i}=f(i_{1},i_{2},\cdots,i_{m})$ for $i= 0,1,\cdots,L-1$, which is a naturally generalization of \cite{CC}. For convenience, we ignore the superscript of \textcolor{black}{$\mathbf{f}^{(L)}$} unless the sequence length is undetermined.

\section{Construction of MOCSs with flexible lengths}
In this section, we present a direct construction of MOCSs with flexible lengths. Before giving the new MOCSs, we introduce the following lemmas.
\begin{lem}\cite{CC3}\label{2}
	For an even integer $q$ and any positive integers $m$,  $ k$ with $k\leq m$, let $v$ be an integer with $0\leq v\leq m-k$, and $\pi$ be a permutation of $\mathbb{N}_m$ satisfying the following three conditions:
	
	$ (1) $ $\pi(m-k+1)<\pi(m-k+2)<\cdots<\pi(m-1)<\pi(m)=m.$
	
	$ (2) $ If $v>0$, then \textcolor{black}{$\mathbb{N}_v=\{\pi(1),\pi(2),\cdots,\pi(v)\}.$}
	
	$ 	(3) $ For all $\alpha=1,2,\cdots,k-1,$ if $\pi(t)<\pi(m-k+\alpha),$ then $\pi(t-1)<\pi(m-k+\alpha)$ where $2\leq t\leq m-k.$\\
Let
	$$f=\frac{q}{2}\sum^{m-k-1}_{s=1}x_{\pi(s)}x_{\pi(s+1)}+\sum^{k}_{\alpha=1}\sum^{m-k}_{s=1}c_{\alpha,s}x_{\pi(m-k+\alpha)}x_{\pi(s)}
	+\sum^{m}_{s=1}c_{s}x_{s}+c_{0},$$
	where $c_{\alpha,s}, c_{s}\in \mathbb{Z}_{q}$. Then the set
	$$\mathcal{F}=\left\{\mathbf{f}+\frac{q}{2}\sum^{k}_{\alpha=1}d_{\alpha}\mathbf{x}_{\pi(m-k+\alpha)}+\frac{q}{2}d_{k+1}\mathbf{x}_{\pi(1)}\mid d_{\alpha}\in \{0,1\}\right\}$$ forms a GCS of size $2^{k+1}$ and length $L=2^{m-1}+\sum^{k-1}_{\alpha=1}a_{\alpha}2^{\pi(m-k+\alpha)-1}+2^v$ with $a_{\alpha}\in \{0,1\}.$
\end{lem}

\begin{lem}\label{5}
	For positive integers $m\geq 2$ and $r<m,$ let $h$ be a bijection  from \textcolor{black}{$S_1=\mathbb{N}_r$} onto \textcolor{black}{$S_2\subseteq \mathbb{N}_m$} with $r$ elements. Suppose that $h(u)$  is the smallest element of $S_2$.
  Let $i$ be an integer with
	
	$$\sum^{r}_{\substack {l=1\\ l\neq u}}a_{l}q^{h(l)-1}\leq i\leq \sum^{r}_{\substack {l=1\\ l\neq u}}a_{l}q^{h(l)-1}+q^{h(u)}-1,$$ where $a_{l}\in \mathbb{Z}^{*}_q$ for $l\in S_1\setminus\{u\}$ and $(i_{1},i_{2},\cdots,i_{m})$ is the $q$-ary representation of $i.$ Also let $i^{(t)}$ be an integer with $q$-ary representation $(i_{1},i_{2},\cdots,i_{k}\oplus t,\cdots,i_{m})$ for positive integers $k\leq h(u)$ and $t\in \mathbb{Z}^*_{q}$. Then we have  $$\sum^{r}_{\substack {l=1\\l\neq u}}a_{l}q^{h(l)-1}\leq i^{(t)}\leq \sum^{r}_{\substack {l=1\\l\neq u}}a_{l}q^{h(l)-1}+q^{h(u)}-1.$$
\end{lem}
\begin{proof}
	For convenience, we let $j=i-\sum\limits^{r}_{\substack {l=1\\ l\neq u}}a_{l}q^{h(l)-1}$ and $(j_{1},j_{2},\cdots,j_{m})$ be the $q$-ary representation of $j.$ Then $0\leq j\leq q^{h(u)}-1$, which means $j_{s}=0$ for $s\geq h(u)+1$. Similarly, we let $j^{(t)}=i^{(t)}-\sum\limits^{r}_{\substack {l=1\\l\neq u}}a_{l}q^{h(l)-1}$ with $q$-ary representation $(j_{1},j_{2},\cdots,j_{k}\oplus t,\cdots,j_{m})$. Obviously, the $q$-ary representation of $j$ differs from that of $j^{(t)}$ in only one position $k$. So we obtain $j^{(t)}_{s}=j_{s}=0$ for $s\geq h(u)+1$ which implies $0\leq j^{(t)}\leq q^{h(u)}-1$. Therefore, $$\sum^{r}_{\substack {l=1\\ l\neq u}}a_{l}q^{h(l)-1}\leq i^{(t)}\leq \sum^{r}_{\substack {l=1\\ l\neq u}}a_{l}q^{h(l)-1}+q^{h(u)}-1.$$
	
\end{proof}

\begin{lem} \label{lem2}
	For positive integers $m\geq 2$ and $r<m,$ let $i$ and $h$ be the same as that of Lemma \ref{5}. If $i\leq \sum\limits^{r}_{\substack {l=1\\ l\neq u}}a_{l}q^{h(l)-1}+q^{h(u)}-1$ and $i_{h(l)}=a_{l}$ for $l\in S_1\setminus\{u\}$. Then we have $i_{s}=0$ for $s=h(u)+1, h(u)+2,\cdots,m-1$ and $s\neq h(l) $ for $l\in S_1\setminus\{u\}$.
\end{lem}
\begin{proof}
	Suppose the conclusion doesn't hold, we assume $i_{t}=b\neq 0$ where $h(u)+1\leq t\leq m-1$ and $t\neq h(l)$ for $l\in S_1\setminus\{u\}$. Then we have $i\geq \sum\limits^{r}_{\substack {l=1\\ l\neq u}}a_{l}q^{h(l)-1}+b
	q^{t-1}\geq \sum\limits^{r}_{\substack {l=1\\ l\neq u}}a_{l}q^{h(l)-1}+q^{h(u)}$ which contradicts the condition.
\end{proof}

\begin{lem}\label{lem1}\cite{CY}
	\textcolor{black}{Let $q$ be an even number, $(i_1,i_2,\cdots,i_m)$  and $(j_1,j_2,\cdots,j_m)$  be the binary representations of $i$ and $j$, respectively, and let $\{I_1, I_2,\cdots,I_d\}$ be a partition of the set $\mathbb{N}_m$. Let $ \pi_{\alpha} $ be a bijection from $\mathbb{N}_{m_{\alpha}}$ to $ I_{\alpha} $, where $ \lvert I_{\alpha}\rvert =m_{\alpha} $ for any $ \alpha\in \mathbb{N}_d. $} If the following three  conditions are satisfied:
	
	$ (1) $ $ \alpha_{1} $ is the largest integer satisfying
	$ i_{\pi_{\alpha}(\beta)}= j_{\pi_{\alpha}(\beta)} $ for
	$ \alpha\in \mathbb{N}_{\alpha_1}$ and $ \beta\in \mathbb{N}_{m_{\alpha}}. $
	
	$ (2) $ $ \beta_1 $ is the smallest integer such that $ i_{\pi_{\alpha_{1}(\beta_{1})}}\neq j_{\pi_{\alpha_{1}(\beta_{1})}}. $
	
	$ (3) $    Let $ i' $ and $ j' $ be integers which differ from $ i $ and $ j $, respectively, in only one position $ \pi_{\alpha_{1}(\beta_{1}-1)}, $ that is, $ i'_{\pi_{\alpha_{1}(\beta_{1}-1)}}=1-i_{\pi_{\alpha_{1}(\beta_{1}-1)}} $ and $ j'_{\pi_{\alpha_{1}(\beta_{1}-1)}}=1-j_{\pi_{\alpha_{1}(\beta_{1}-1)}} $.\\
	Then $$ f_{n,i}-f_{n,j}-f_{n,i'}+f_{n,j'}\equiv \frac{q}{2}\pmod q,$$
	\textcolor{black}{where $f(x)$ as shown in Eq. (1) of \cite{CY}.}
\end{lem}
\textcolor{black}{\begin{lem}\label{lem3}
		Let $ \textbf{x}_{n_1}, \textbf{x}_{n_2},\cdots,\textbf{x}_{n_d}$ be the sequences corresponding to EBFs $ x_{n_1}, x_{n_2},\cdots,x_{n_d}$, respectively, where $ n_1<n_2<\cdots<n_d. $ Let   $ \textbf{u}=(u_0,u_1,\cdots,u_{L-1})=a_1\textbf{x}_{n_1}\oplus a_2\textbf{x}_{n_2}\oplus \cdots \oplus a_d\textbf{x}_{n_d} $ be a $q$-ary sequence with $ a_i\in \mathbb{Z}_q $ for any $ i\in \mathbb{N}_d, $ which is a   linear combination of $ \textbf{x}_{n_1}, \textbf{x}_{n_2},\cdots,\textbf{x}_{n_d}.$  If $ q^{n_1}\mid L $ and $a_1\neq 0$, let $(i_1,i_2,\cdots,i_m)$ be the binary representation of $i$, and let   $i^{(t)} $ differ from $i$ in only one position $n_1$, i.e.,   $$(i^{(t)}_1,i^{(t)}_2,\cdots,i^{(t)}_{m})=(i_1,i_2,\cdots,i_{n_1-1},i_{n_1}\oplus t,i_{n_1+1},\cdots,i_{m}),$$ where $t\in \mathbb{Z}^*_{q}$. Then $$\xi^{u_i}+\xi^{u_{i^{(1)}}}+\xi^{u_{i^{(2)}}}+\cdots+\xi^{u_{i^{(q-1)}}}=0.$$
	\end{lem}
\begin{proof}
	Since $ q^{n_1}\mid L $, then for any integer $i$,  $$u_{i^{(t)}}-u_i=(a_1i^{(t)}_{n_1}\oplus a_2i^{(t)}_{n_2}\oplus\cdots a_ki^{(t)}_{n_k} )-(a_1i_{n_1}\oplus a_2i_{n_2}\oplus\cdots a_ki_{n_k} )\equiv a_1(i^{(t)}_{n_1}-i_{n_1})\equiv a_1t\pmod q,$$
Thus we have $$1+\xi^{u_{i^{(1)}}-u_i}+\xi^{u_{i^{(2)}}-u_i}+\cdots+\xi^{u_{i^{(q-1)}}-u_i}=0.$$
\end{proof}}
Now we state our construction in the following Theorem \ref{thm1}, which is based on Lemma \ref{2}.
\begin{thm}\label{thm1}
	Let $m$, $ d $, \textcolor{black}{$d'$}, $ v $ be positive integers with $ 0<d'< d<m $ and $ v<m $. Let \{$I_1$,  $ I_2$,  $\cdots$,  $I_d$\}  be a partition of the set $\mathbb{N}_{m-v}$. Put $ \pi_{\alpha} $ be a bijection from $ \mathbb{N}_{m_{\alpha}}$ to $ I_{\alpha} $, where $ \lvert I_{\alpha}\rvert =m_{\alpha} $ for any $ \alpha\in \mathbb{N}_d. $ Let $ u $ be an integer with \textcolor{black}{$ \sum_{{\alpha}=1}^{d'}m_{\alpha}< u< \sum_{\alpha=1}^{d'+1}m_{\alpha}$},  we  impose an additional condition below:
	\textcolor{black}{$$\{\pi_{1}(1),\pi_{1}(2),\cdots,\pi_{1}(m_1),\pi_{2}(1),\cdots,\pi_{d'+1}(u')\}=\{1,2.\cdots,u\},$$ where $0<u'< m_{d'+1}$.}
	Let $(n_{1},n_{2},\cdots,n_{d+v})$ and $(p_{1},p_{2},\cdots,p_{d'})$ be the $q$-ary representations of $n$ and $p$, respectively. Let
	\begin{align}
		f(x)&=\sum^{d}_{\alpha=1}\sum^{m_{\alpha}-1}_{\beta=1}a_{\alpha,\beta}x_{\pi_{\alpha}(\beta)}x_{\pi_{\alpha}(\beta+1)}+\sum^{d}_{\alpha=1}\sum^{m_{\alpha}}_{\beta=1}\sum_{k=1}^{v}b_{\alpha,\beta,k}x_{\pi_{\alpha}(\beta)}x_{m-v+k}+\sum_{l=1}^{q-1}\sum_{s=1}^{m}c_{s,l}x^l_s+c_0,\\
		f^p_{n}(x) &= f(x)
		+\sum^{d}_{\alpha=1}n_{\alpha} x_{\pi_{\alpha}(1)}+\sum^{v}_{k=1}n_{k+d}x_{m-v+k}
		+c\sum^{d'}_{\alpha=1}p_{\alpha}x_{\pi_{\alpha}(m_{\alpha})},
	\end{align}
	where $a_{\alpha,\beta}, c\in \mathbb{Z}^{*}_{q}$ are co-prime with $q$
	and $b_{\alpha,\beta,k},c_{s,l},c_0\in \mathbb{Z}_{q}$. \textcolor{black}{Then $\{\mathcal{F}^{0},\mathcal{F}^{1},\cdots,\mathcal{F}^{q^{d'}-1}\}$ generates a $(q^{d'},q^{v+d},L)$-MOCS}
	with $L=a_mq^{m-1}+\sum^{v-1}_{k=1}a_{k}q^{m-v+k-1}+q^{u}$, $ a_{k}\in \mathbb{Z}_q$ and  $a_m\in \mathbb{Z}^*_q$, where \textcolor{black}{$\mathcal{F}^p=\{\mathbf{f}^p_{0},\mathbf{f}^p_{1},\cdots,\mathbf{f}^p_{q^{v+d}-1}\}.$}
\end{thm}
\begin{proof}
\textcolor{black}{Since for sequences $\mathbf{f}^{p}_{n} $ and $\mathbf{f}^{p'}_{n}$, $R_{\mathbf{f}^{p'}_{n},\mathbf{f}^{p}_{n}}(-\tau)=R^{*}_{\mathbf{f}^{p}_{n},\mathbf{f}^{p'}_{n}}(\tau)$, then it suffice to prove that for $0\leq p, p' \leq q^{d'}-1$ and $0<\tau\leq L-1$,
	$$R_{\mathcal{F}^{p},\mathcal{F}^{p'}}(\tau)=\sum^{q^{v+d}-1}_{n=0}\sum^{L-1-\tau}_{i=0}\xi^{f^{p}_{n,i}-f^{p'}_{n,i+\tau}}
	=\sum^{L-1-\tau}_{i=0}
	\sum^{q^{v+d}-1}_{n=0}\xi^{f^{p}_{n,i}-f^{p'}_{n,i+\tau}}=0,$$ where $ f^{p}_{n,i} $ and $ f^{p'}_{n,j} $ are the $ (i+1) $-th and the $ (j+1) $-th
	element of sequence $ \mathbf{f}^{p}_{n} $ and $ \mathbf{f}^{p'}_{n} $, respectively. }	
	For simplicity, we assume $a_k\neq 0$ for any \textcolor{black}{$k\in \mathbb{N}_{v-1}$}.
	Throughout this paper, for a given integer $ i $, we set $ j=i+\tau $ and let $ (i_1,i_2,\cdots,i_m) $ and $ (j_1,j_2,\cdots,j_m) $ be the $ q $-ary representations of $ i $ and $ j $, respectively. Let $(p_{1},p_{2},\cdots,p_{d'})$ and $(p'_{1},p'_{2},\cdots,p'_{d'})$ are the $q$-ary representations of $p$ and $p'$, respectively.

	Case 1: If $i_{\pi_{\alpha}(1)}\neq j_{\pi_{\alpha}(1)}$ for some  $\alpha\in \mathbb{N}_{d}$ or $i_{m-v+k}\neq j_{m-v+k}$ for some   $k \in \mathbb{N}_{v}.$ Then
	\begin{align*}
		R_{\mathcal{F}^{p},\mathcal{F}^{p'}}(\tau)=\sum_{i=0}^{L-1-\tau}\xi^{f_i-f_j}\prod_{\alpha=1}^{d}\left( \sum_{n_{\alpha}=0}^{q-1}\xi^{n_{\alpha}\left( i_{\pi_{\alpha}(1)}- j_{\pi_{\alpha}(1)}\right)}\right)\prod_{\alpha=1}^{d'}\xi^{ p_{\alpha}i_{\pi_{\alpha}(m_{\alpha})}- p'_{\alpha}j_{\pi_{\alpha}(m_{\alpha})} }A=0.
	\end{align*}
	where $ A=\prod_{k=1}^{v}\left( \sum_{n_{d+k}=0}^{q-1}\xi^{n_{d+k}(i_{m-v+k}- j_{m-v+k})}\right) $.
	
	Case 2: If $i_{\pi_{\alpha}(1)}= j_{\pi_{\alpha}(1)}$ for all  $\alpha\in \mathbb{N}_{d}$, $i_{m-v+k}= j_{m-v+k}$ for all $k \in\mathbb{N}_v$, and $ i_m=j_m=0. $ 
	Then \begin{align*}
		R_{\mathcal{F}^{p},\mathcal{F}^{p'}}(\tau)=q^{d+v}\sum_{i=0}^{L-1-\tau}\xi^{f_i-f_j}\prod_{\alpha=1}^{d'}\xi^{ p_{\alpha}i_{\pi_{\alpha}(m_{\alpha})}- p'_{\alpha}j_{\pi_{\alpha}(m_{\alpha})} }.
	\end{align*}
Similar to Lemma \ref{lem1}, we assume that
	
	(1) $\alpha_{1} $ is the largest integer satisfying
	$ i_{\pi_{\alpha}(\beta)}= j_{\pi_{\alpha}(\beta)} $ for
	$ \alpha\in \mathbb{N}_{\alpha_1}$ and $ \beta\in \mathbb{N}_{m_{\alpha}}$.
 
 (2) $ \beta_1 $ is the smallest integer such that $ i_{\pi_{\alpha_{1}}(\beta_{1})}\neq j_{\pi_{\alpha_{1}}(\beta_{1})} $.

(3) Let $ i^{(t)} $ and $ j^{(t)} $ be integers which differ from $ i $ and $ j $, respectively, in only one position $ \pi_{\alpha_{1}}(\beta_{1}-1), $ that is, $ i^{(t)}_{\pi_{\alpha_{1}}(\beta_{1}-1)}=t\oplus i_{\pi_{\alpha_{1}}(\beta_{1}-1)} $ and $ j^{(t)}_{\pi_{\alpha_{1}}(\beta_{1}-1)}=t\oplus j_{\pi_{\alpha_{1}}(\beta_{1}-1)} $.
	
	Thus we get
	$$ f_{i^{(t)}}-f_i-f_{j^{(t)}}+f_j=ta_{\alpha_1,\beta_{1}-1}\left( i_{\pi_{\alpha_1}(\beta_{1})}- j_{\pi_{\alpha_1}(\beta_{1})}\right)  $$ and
	$$ \xi^{f_i-f_j}+\xi^{f_{i^{(1)}}-f_{j^{(1)}}}+\xi^{f_{i^{(2)}}-f_{j^{(2)}}}+\cdots+\xi^{f_{i^{(q-1)}}-f_{j^{(q-1)}}}=0,$$
	which implies   \begin{align*}
		R_{\mathcal{F}^{p_1},\mathcal{F}^{p_2}}(\tau)=q^{d+v}\sum_{i=0}^{L-1-\tau}\xi^{f_i-f_j}\prod_{\alpha=1}^{d'}\xi^{ p_{\alpha}i_{\pi_{\alpha}(m_{\alpha})}- p'_{\alpha}j_{\pi_{\alpha}(m_{\alpha})} }=0.
	\end{align*}
	
	Case 3: If $i_{\pi_{\alpha}(1)}= j_{\pi_{\alpha}(1)}$ for all  $\alpha \in \mathbb{N}_{d}$, $i_{m-v+k}= j_{m-v+k}$ for all $k \in \mathbb{N}_{v}$, and $ i_m=j_m=a_m\neq 0. $ Suppose $k_1$ is the largest integer such that $i_{m-v+k}=j_{m-v+k}=0$ for $k <v,$ i.e., $i_{m-v+k}=j_{m-v+k}=a_k\neq 0$ for $ k\in \{k_1+1,k_1+2,\cdots,v\}, $ then
	\begin{align*}
		i,j&<L=a_mq^{m-1}+\sum^{v-1}_{\alpha=1}a_{k}q^{m-v+k-1} +q^{u}\\
		&\leq a_mq^{m-1} +\sum^{v-1}_{k=k_1+1}a_{k}q^{m-v+k-1}+q^{m-v+k_1-1}-1.
	\end{align*}
	According to Lemma \ref{5}  and $ \pi_{\alpha_{1}(\beta_{1}-1)}<{m-v+k_1-1}, $ we have
	$$i^{(t)},j^{(t)}\leq a_mq^{m-1} +\sum^{v-1}_{k=k_1+1}a_{k}q^{m-v+k-1}+q^{m-v+k_1-1}-1< L.$$ Therefore, we get
	
	$$\xi^{f_{i}-f_{j}}+\xi^{f_{i^{(1)}}-f_{j^{(1)}}}+\cdots+\xi^{f_{i^{(q-1)}}-f_{j^{(q-1)}}}=0.$$

	Case 4: $i_{\pi_{\alpha}(1)}= j_{\pi_{\alpha}(1)}$ for all  $\alpha \in \mathbb{N}_d$, $i_{m-v+k}= j_{m-v+k}$ for all $k\in \mathbb{N}_{v}$, and $ i_m=j_m=a_m\neq 0. $ We also consider that $i_{m-v+k}=j_{m-v+k}=a_k\neq 0 $ for all $k\in \mathbb{N}_{v},$
	$$i,j<L=a_{m}q^{m-1}+\sum^{v-1}_{k=1}a_{k}q^{m-v+k-1}+q^{u}.$$ According to Lemma \ref{lem2}, we have $i_s=j_s=0$ for $s=u+1, u+2, \cdots, m-v-1$, so $\pi_{\alpha_{1}}(\beta_{1})\leq u, $ and $\pi_{\alpha_{1}}(\beta_{1}-1)\leq u. $
	  Therefore,
	$$i^{(t)},j^{(t)}\leq a_{m}q^{m-1}+\sum^{v-1}_{k=1}a_{k}q^{m-v+k-1}+q^{u}< L$$ and
	$$\xi^{f_{i}-f_{j}}+\xi^{f_{i^{(1)}}-f_{j^{(1)}}}+\cdots+\xi^{f_{i^{(q-1)}}-f_{j^{(q-1)}}}=0.$$
	
	Combining the above four cases, we can conclude that $R_{\mathcal{F}^{p},\mathcal{F}^{p'}}(\tau)=0$ for $0<\tau\leq L-1$.
	
	 Next, it remains to show that for $0\leq p\neq p' \leq q^{d'}-1$, $$R_{\mathcal{F}^{p},\mathcal{F}^{p'}}(0)=\sum^{q^{v+d}-1}_{n=0}\sum^{L-1}_{i=0}{\xi}^{f^{p}_{n,i}-f^{p'}_{n,i}}=0.$$
	Since $p\neq p'$, there exists a smallest $s\in \mathbb{N}_{d'}$ such that $p_{s}\neq p'_{s}$. \textcolor{black}{Then according to Lemma \ref{lem3},
		for any $0\leq i\leq L-1$, there exists $i^{(t)}$ whose $q$-ary representation differs from $i$ in only one position $s$, i.e., $(i_1,i_2,\cdots,i_{s-1},i_s\oplus t,i_{s+1},\cdots,i_m)$ for any $t\in \mathbb{Z}^{*}_{q}$. Therefore, we get
		\begin{align*}
			&\xi^{f^{p}_{n,i}-f^{p'}_{n,i}}+
			\xi^{f^{p}_{n,i^{(1)}}-f^{p'}_{n,i^{(1)}}}+ \cdots +
			\xi^{f^{p}_{n,i^{(q-1)}}-f^{p'}_{n,i^{(q-1)}}}\\
			=&\xi^{f^{p}_{n,i}-f^{p'}_{n,i}}\left( 1+\xi^{f^{p}_{n,i^{(1)}}-f^{p'}_{n,i^{(1)}}-f^{p}_{n,i}+f^{p'}_{n,i}}+\cdots +
			\xi^{f^{p}_{n,i^{(q-1)}}-f^{p'}_{n,i^{(q-1)}}-f^{p}_{n,i}+f^{p'}_{n,i}}\right)\\
			=& \xi^{f^{p}_{n,i}-f^{p'}_{n,i}}\left(1+\xi^{c(p_{s}-p'_{s})}+\cdots+\xi^{c(p_{s}-p'_{s})(q-1)} \right)\\
			=&0
		\end{align*}
		and
		$$R_{\mathcal{F}^{p},\mathcal{F}^{p'}}(0)=\sum^{q^{k+1}-1}_{n=0}\sum^{L-1}_{i=0}{\xi}^{f^{p}_{n,i}-f^{p'}_{n,i}}=0.$$
	}
	By the above discussion, we obtain that  $\{\mathcal{F}^{p}\mid p\in \{0,1,\cdots,q^{d'}-1\} \}$ is a $(q^{d'},q^{v+d},L)$-MOCS with $L=a_mq^{m-1}+\sum^{v-1}_{k=1}a_{k}q^{m-v+k-1}+q^{u}$, where $ a_{k}\in \mathbb{Z}_q$ and  $a_m\in \mathbb{Z}^*_q$.
	
\end{proof}
\begin{rem}
	In Theorem \ref{thm1}, if we let $q=2$ and all $a_{k}=0$ and $a_m=1$, then the length $L=a_mq^{m-1}+\sum^{v-1}_{k=1}a_{k}q^{m-v+k-1}+q^{u}$ turns into the form $2^{m-1}+2^{u},$ this result is coveblack in \cite{SW}.
\end{rem}

\textcolor{black}{\begin{ex}
			Let $m=5, $  $v=1, $  $d=2, $ $d'=1,$ $m_1=m_2=2, $  $(\pi_{1}(1),\pi_{1}(2),\pi_{2}(1),\pi_{2}(2))=(1,2,3,4)  $ and all $a_{\alpha,\beta}, b_{\alpha,\beta,k}, c_{s,l}, c_0,c$ are equal to 1. Then   $\{\mathcal{F}^{0},\mathcal{F}^{1},\mathcal{F}^{2}\}$ forms a  ternary $(3,27,108)$-MOCS from Theorem \ref{thm1}.	
\end{ex}}

\textcolor{black}{\section{Constructions of CCCs and optimal ZCCSs}}
In this section, we mainly propose an approach to constructing an optimal ZCCS. Before doing this work, we need to construct  CCCs as a preparing work.

\begin{thm}\label{8}
	Let $m$, $ d $ be positive integers with $ 2\leq d<m $, and $\{I_1, I_2,\cdots,I_d\}$ be a partition of the set $\mathbb{N}_m$. Put $ \pi_{\alpha} $ be a bijection from $ \mathbb{N}_{m_{\alpha}}$ to $ I_{\alpha} $, where $ \lvert I_{\alpha}\rvert =m_{\alpha} $ for any $ \alpha\in \{1,2,\cdots,d\}. $ Let
	\begin{align*}
		f(x)=&\sum_{\alpha=1}^{d}\sum_{\beta=1}^{m_{\alpha}-1}a_{\alpha,\beta}x_{\pi_{\alpha}(\beta)}x_{\pi_{\alpha}(\beta+1)}+\sum_{l=1}^{q-1}\sum_{u=1}^{m}h_{u,l}x^l_u+h_0,\\
		f^p_{n}(x)=&f(x)+\sum_{\alpha=1}^{d}n_{\alpha}x_{\pi_{\alpha}(1)}+\sum_{\alpha=1}^{d}p_{\alpha}x_{\pi_{\alpha}(m_{\alpha})},
	\end{align*}
	where $a_{\alpha,\beta} \in \mathbb{Z}^{*}_{q}$ is co-prime with $q$, $h_{u,l}$, $ h_0 \in \mathbb{Z}_{q}$,  $(n_{1},n_{2},\cdots,n_{d})$ and $(p_{1},p_{2},\cdots,p_{d})$ are the $q$-ary representations of $n$ and $ p $, respectively. Then the set \textcolor{black}{$ \{\mathcal{F}^{0},\mathcal{F}^{1},\cdots,\mathcal{F}^{q^d-1}\} $} forms a $q$-ary CCC   with
	\textcolor{black}{$\mathcal{F}^p=\{\mathbf{f}^p_{0},\mathbf{f}^p_{1},\cdots,\mathbf{f}^p_{q^d-1}\}$. }
\end{thm}

\begin{proof}
The proof consists of two parts. In the first part, \textcolor{black}{we demonstrate that  for any $ 0 \leq p, p'\leq q^d-1 $ and $0<\tau\leq q^m-1$, $\mathcal{F}^{p}$  and $\mathcal{F}^{p'}$ satisfy the ideal correlation property, i.e.,
	$$R_{\mathcal{F}^{p},\mathcal{F}^{p'}}(\tau)=\sum_{n=0}^{q^{d}-1}R_{\mathbf{f}^{p}_n,\mathbf{f}^{p'}_n}(\tau)=\sum_{n=0}^{q^{d}-1}\sum_{i=0}^{q^m-1-\tau}\xi^{f^{p}_{n,i}-f^{p'}_{n,j}}=\sum_{i=0}^{q^m-1-\tau}\sum_{n=0}^{q^{d}-1}\xi^{f^{p}_{n,i}-f^{p'}_{n,j}}=0, $$ where $ f^{p}_{n,i} $ and $ f^{p'}_{n,j} $ are the $ (i+1) $-th and the $ (j+1) $-th
	element of sequence $ \mathbf{f}^{p}_{n} $ and $ \mathbf{f}^{p'}_{n} $, respectively.}  Similarly, let the definitions of $i, j, i^{(t)} $ and $ j^{(t)} $  be given as Theorem \ref{thm1}. 	Furthermore, we divide the set $ \{i \mid 0\leq i\leq  q^m-1-\tau \}$ into two parts: $ S_1(\tau)=\{i\mid \exists\ \alpha\in \{1,2,\cdots,d\},\ 0\leq i\leq q^m-1-\tau,\ i_{\pi_{\alpha}(1)}\neq j_{\pi_{\alpha}(1)}\} $  and  $ S_2(\tau)=\{i\mid \forall\ \alpha\in \{1,2,\cdots,d\},\ 0\leq i\leq q^m-1-\tau,\ i_{\pi_{\alpha}(1)}= j_{\pi_{\alpha}(1)}\} $.
 Thus we obtain that
\begin{align*}
	R_{\mathcal{\textcolor{black}{F}}^{p},\mathcal{\textcolor{black}{F}}^{p'}}(\tau)=&\sum_{i=0}^{q^m-1-\tau}\sum_{n=0}^{q^{d}-1}\xi^{f^{p}_{n,i}-f^{p'}_{n,j}}\\
	=&\sum_{i=0}^{q^m-1-\tau}\xi^{f_i-f_j}\prod_{\alpha=1}^{d}\left( \sum_{n_{\alpha}=0}^{q-1}\xi^{n_{\alpha}\left( i_{\pi_{\alpha}(1)}- j_{\pi_{\alpha}(1)}\right) }\right)\prod_{\alpha=1}^{d}\xi^{p_{\alpha}i_{\pi_{\alpha}(m_{\alpha})}- p'_{\alpha}j_{\pi_{\alpha}(m_{\alpha})}}\\
	=&\sum_{i\in S_1(\tau)}\xi^{f_i-f_j}\prod_{\alpha=1}^{d}\left( \sum_{n_{\alpha}=0}^{q-1}\xi^{n_{\alpha}\left( i_{\pi_{\alpha}(1)}- j_{\pi_{\alpha}(1)}\right)}\right)\prod_{\alpha=1}^{d}\xi^{p_{\alpha}i_{\pi_{\alpha}(m_{\alpha})}- p'_{\alpha}j_{\pi_{\alpha}(m_{\alpha})}}\\&+\sum_{i\in S_2(\tau)}\xi^{f_i-f_j}\prod_{\alpha=1}^{d}\left( \sum_{n_{\alpha}=0}^{q-1}\xi^{n_{\alpha}\left( i_{\pi_{\alpha}(1)}- j_{\pi_{\alpha}(1)}\right)}\right)\prod_{\alpha=1}^{d}\xi^{p_{\alpha}i_{\pi_{\alpha}(m_{\alpha})}- p'_{\alpha}j_{\pi_{\alpha}(m_{\alpha})}}\\
	=&q^d\sum_{i\in S_2(\tau)}\xi^{f_i-f_j}\prod_{\alpha=1}^{d}\xi^{p_{\alpha}i_{\pi_{\alpha}(m_{\alpha})}- p'_{\alpha}j_{\pi_{\alpha}(m_{\alpha})}},
\end{align*}
where  $ (p_{k,1},p_{k,2},\cdots,p_{k,d})  $ is the $ q $-ary representation of $ p_k $ for any $ k\in \{1,2\}. $
For any $ i\in S_2(\tau) $, according to the Case 2 of first part in Theorem \ref{thm1}, we have
$$ f_{i^{(t)}}-f_i-f_{j^{(t)}}+f_j=ta_{\alpha_1,\beta_{1}-1}\left( i_{\pi_{\alpha_1}(\beta_{1})}- j_{\pi_{\alpha_1}(\beta_{1})}\right)  $$ and
$$ (\xi^{f_i-f_j}+\xi^{f_{i^{(1)}}-f_{j^{(1)}}}+\xi^{f_{i^{(2)}}-f_{j^{(2)}}}+\cdots+\xi^{f_{i^{(q-1)}}-f_{j^{(q-1)}}})\prod_{\alpha=1}^{d}\xi^{p_{\alpha}i_{\pi_{\alpha}(m_{\alpha})}- p'_{\alpha}j_{\pi_{\alpha}(m_{\alpha})}}=0. $$

According to the above discussion, we know that the ideal correlation property is available for any $ \tau>0$.  Now, we need to prove that for any  $ 0 \leq p\neq p'\leq q^d-1 $ and $ \tau=0$,
	\begin{align*}
		R_{\mathcal{\textcolor{black}{F}}^{p},\mathcal{\textcolor{black}{F}}^{p'}}(0)
		=\sum_{n=0}^{q^{d}-1}R_{\mathbf{\textcolor{black}{f}}^{p}_{n},\mathbf{\textcolor{black}{f}}^{p'}_{n}}(0)
		=\sum_{n=0}^{q^{d}-1}\sum_{i=0}^{q^m-1}\xi^{\sum_{\alpha=1}^{d}(p_{\alpha}\oplus p'_{\alpha})i_{\pi_{\alpha}(m_\alpha)}}
		=0.
	\end{align*}
	Put $ \textbf{d}=\sum_{\alpha=1}^{d}(p_{\alpha}\oplus p'_{\alpha})\textbf{x}_{\pi_{\alpha}(m_{\alpha})} $. Due to  each $ \textbf{x}_{\pi_{\alpha}(m_{\alpha})} $ is a balanced sequence, the linear combination of  $\textbf{x}_{\pi_{1}(m_{1})},\textbf{x}_{\pi_{2}(m_{2})},\cdots,\textbf{x}_{\pi_{d}(m_{d})} $ is balanced, i.e., $ \textbf{d} $ is balanced. Then we have
	\begin{align*}
	R_{\mathcal{	\textcolor{black}{F}}^{p},\mathcal{\textcolor{black}{F}}^{p'}}(0)=
		\sum_{n=0}^{q^{d}-1}\sum_{i=0}^{q^m-1}\xi^{\sum_{\alpha=1}^{d}(p_{\alpha}\oplus p'_{\alpha})i_{\pi_{\alpha}(m_\alpha)}}
		=0,
	\end{align*}
	which completes the proof.

\end{proof}

With the help of the above Theorem \ref{8}, the following $(q^{v+d},q^{d},q^{m},q^{m-v})$-ZCCSs can be obtained easily.
\begin{thm} \label{4}
	Let $m$, $ d $, $ v $ be positive integers with $d\leq m-v $ and $ v<m $. Let $\{I_1, I_2,\cdots,I_d\}$  be a partition of the set $\mathbb{N}_{m-v}$. Put $ \pi_{\alpha} $ be a permutation from $\mathbb{N}_{m_{\alpha}}$ to $ I_{\alpha} $, where $ \lvert I_{\alpha}\rvert =m_{\alpha} $ for any $ \alpha\in \mathbb{N}_{d}. $ Also let 	
	\begin{align*}
		f(x)=&\sum_{\alpha=1}^{d}\sum_{\beta=1}^{m_{\alpha}-1}a_{\alpha,\beta}x_{\pi_{\alpha}(\beta)}x_{\pi_{\alpha}(\beta+1)}+\sum_{l=1}^{q-1}\sum_{u=1}^{m}h_{u,l}x^l_u+h_0,\\
		f^p_{n}(x)=&f(x)+\sum_{\alpha=1}^{d}n_{\alpha}x_{\pi_{\alpha}(1)}+b\left( \sum_{\alpha=1}^{d}p_{\alpha}x_{\pi_{\alpha}(m_{\alpha})}+\sum_{k=1}^{v}p_{k+d}x_{m-v+k}\right) ,
	\end{align*}
	where $(n_{1},n_{2},\cdots,n_{d})$ and $(p_{1},p_{2},\cdots,p_{v+d})$ are the $q$-ary representations of $n$ and $p$,  respectively,	$a_{\alpha,\beta}, b\in \mathbb{Z}^{*}_{q}$ are both co-prime with $q$, and $h_{u,l}, h_0 \in \mathbb{Z}_{q}$.    Then \textcolor{black}{$\left\lbrace \mathcal{F}^{0},\mathcal{F}^{1},\cdots,\mathcal{F}^{q^{v+d}-1}\right\rbrace $} forms a $(q^{v+d},q^{d},q^{m},q^{m-v})$-ZCCS with \textcolor{black}{$\mathcal{F}^{p}=\left\lbrace \mathbf{f}^{p}_{0},\mathbf{f}^{p}_{1},\cdots,\mathbf{f}^{p}_{q^d-1}\right\rbrace $.}
\end{thm}
\begin{proof}
	It is obvious that every sequence \textcolor{black}{$\mathbf{f}^{p}_{n}$} can be divided into $q^{v}$ relevant sub-sequence by a concatenate method, i.e.,
	\textcolor{black}{$$\mathbf{f}^{p}_{n}=\mathbf{g}^{p}_{n,0}|\mathbf{g}^{p}_{n,1}|\cdots|\mathbf{g}^{p}_{n,q^{v}-1},$$ }
	\textcolor{black}{Each $\mathbf{g}^{p}_{n,e}$ can be expressed as $\mathbf{g}^{p}_{n,0}\oplus x_e$, i.e., $\mathbf{g}^{p}_{n,e}=\mathbf{g}^{p}_{n,0}\oplus x_e,$ where $ \mathbf{g}^{p}_{n,e} $ denotes the $ (e+1) $-th sub-sequence of $ \mathbf{f}^{p}_{n} $, $x_e\in \mathbb{Z}_{q}$  and $e\in\{0,1,2,\cdots,q^{v}-1\}.$}  For
	any $0<\tau\leq q^{m-v}-1$ and any $0\leq p\leq q^{v+d}-1$,
	\begin{align*}
		R_{\mathcal{F}^p}(\tau) &=\sum^{q^{d}-1}_{n=0}R_{\mathbf{f}^p_{n}}(\tau)\\
		&=\left(1+\sum^{q^{v}-1}_{k=1}\xi^{u_{k}-w_{k}}\right)\sum^{q^{d}-1}_{n=0}R_{\mathbf{g}^p_{n,0}}(\tau)+
		\left(\xi^{-w_{1}}+\sum^{q^{v}-2}_{k=1}\xi^{u_{k}-w_{k+1}}\right)\sum^{q^{d}-1}_{n=0}R^{*}_{\mathbf{g}^p_{n,0}}(q^{v}-\tau)\\
		&=0.
	\end{align*}
	By the way of Theorem \ref{8}, we conclude that the sequence set \textcolor{black}{$\left\lbrace \mathbf{g}^{p}_{0,0},\mathbf{g}^{p}_{1,0},\cdots,\mathbf{g}^{p}_{q^{d}-1,0}\right\rbrace $} forms a GCS. Therefore, we know that  \textcolor{black}{$\left\lbrace \mathbf{f}^{p}_{0},\mathbf{f}^{p}_{1},\cdots,\mathbf{f}^{p}_{q^d-1}\right\rbrace $} satisfies the auto-correlation property for $0<\tau\leq q^{m-v}-1$.
	
	Next, we verify the cross-correlation property, i.e., for $0\leq p\neq p'\leq q^{v+d}-1$ and for any $ 0<\tau<q^{m-v} $,
	\begin{align*}
		&R_{\mathcal{F}^{p},\mathcal{F}^{p'}}(\tau)\\ &=\sum^{q^{d}-1}_{n=0}R_{\mathbf{f}^{p}_{n},\mathbf{f}^{p'}_{n}}(\tau)\\
		&=\left(1+\sum^{q^{v}-1}_{k=1}\xi^{u_{k}-w_{k}}\right)\sum^{q^{d}-1}_{n=0}R_{\mathbf{g}^{p}_{n,0},\mathbf{g}^{p'}_{n,0}}(\tau)+
		\left(\xi^{-w_{1}}+\sum^{q^{v}-2}_{k=1}\xi^{u_{k}-w_{k+1}}\right)\sum^{q^{d}-1}_{n=0}R^{*}_{\mathbf{g}^{p'}_{n,0},\mathbf{g}^{p}_{n,0}}(q^{v}-\tau)\\
		&=0,
	\end{align*}
	where \textcolor{black}{$\mathbf{f}^{p}_{n}=\mathbf{g}^{p}_{n,0}|(\mathbf{g}^{p}_{n,0}\oplus u_{1})|\cdots|(\mathbf{g}^{p}_{n,0}\oplus u_{q^{v}-1})$} and \textcolor{black}{$\mathbf{f}^{p'}_{n}=\mathbf{g}^{p'}_{n,0}|(\mathbf{g}^{p'}_{n,0}\oplus w_{1})|\cdots|(\mathbf{g}^{p'}_{n,0}\oplus w_{q^{v}-1})$} with $u_{i}, w_{i}\in \mathbb{Z}_{q}.$ The $q$-ary representations of $p$ and $p'$ are $(p_{1},p_{2},\cdots,p_{v+d})$ and $(p'_{1},p'_{2},\cdots,p'_{v+d})$, respectively.
	
	According to the definition of $ f^{p}_{n}(x) $, we get that
	$$ g^{p}_{n,0}(x)=h(x)+\sum_{\alpha=1}^{d}n_{\alpha}x_{\pi_{\alpha}(1)}+b\sum_{\alpha=1}^{d}n_{\alpha}x_{\pi_{\alpha}(m_{\alpha})}, $$ where $ h(x)=\sum_{\alpha=1}^{d}\sum_{\beta=1}^{m_{\alpha}-1}a_{\alpha,\beta}x_{\pi_{\alpha}(\beta)}x_{\pi_{\alpha}(\beta+1)}+\sum_{l=1}^{q-1}\sum_{u=1}^{m-v}h_{u,l}x^l_u+h_0 $ with $\{I_1, I_2,\cdots,I_d\}$ a partition of the set $\mathbb{N}_{m-v}$. Obviously, according to Theorem \ref{8}, we get that $$\sum^{q^{d}-1}_{n=0}R_{\mathbf{g}^{p}_{n,0},\mathbf{g}^{p'}_{n,0}}(\tau)=0  $$ and $$ \sum^{q^{d}-1}_{n=0}R^{*}_{\mathbf{g}^{p'}_{n,0},\mathbf{g}^{p}_{n,0}}(q^{v}-\tau)=0. $$ This shows that $ R_{\mathcal{F}^{p},\mathcal{F}^{p'}}(\tau)=0. $ Similarly, we can prove that $ R_{\mathcal{F}^{p},\mathcal{F}^{p'}}(\tau)=0 $ for any $ -q^d+1\leq \tau<0. $
	
	When $ \tau=0, $ for any $0\leq p\neq p'\leq q^{v+d}-1$,
	\begin{align*}
		R_{\mathcal{F}^{p},\mathcal{F}^{p'}}(0)=\sum^{q^{d}-1}_{n=0}\sum_{i=0}^{q^m-1}\prod_{\alpha=1}^{d}\xi^{b(p_{\alpha}\oplus p'_{\alpha}) i_{\pi_{\alpha}(m_{\alpha})}}\prod_{k=1}^{v}\xi^{b(p_{k+d}\oplus p'_{k+d})i_{m-v+k}}=0.
	\end{align*}
	The equality holds because $ p\neq p'  $ leads to the existence of at least one index $ s\in \mathbb{N}_{v+d} $ such that $ p_{s}\neq p'_{s} $ and $ gcd(b,q)=1. $
	By the above two cases, we get that $ R_{\mathbf{f}^{p},\mathbf{f}^{p'}}(\tau)=0 $ for any $ -q^d<\tau<q^d $ and $ 0\leq p\neq p'\leq q^{v+d}. $ Thus we prove that $ \left\lbrace \mathcal{F}^{0},\mathcal{F}^{1},\cdots,\mathcal{F}^{q^{v+d}-1}\right\rbrace  $ is a $ (q^{v+d},q^{d},q^{m},q^{m-v}) $-ZCCS with $ \mathcal{F}^p=\left\lbrace \mathbf{f}_{0}^{p},\mathbf{f}_{1}^{p},\cdots,\mathbf{f}_{q^d-1}^{p}\right\rbrace  $.
	
\end{proof}

\begin{rem}
	According to Lemma \ref{15}, we know the ZCCS constructed from Theorem \ref{4} is optimal since $M/N=q^{v+d}/q^d=L/Z$ is available. In particular, when $v=0$, the Theorem \ref{4} changes into Theorem \ref{8}.
\end{rem}
\begin{ex}\label{7}
	Let $a_{1,1}=b=1$,  $q=4 $,  $m=3 $,  $v=1 $,  $d=1 $,  $m_1=2 $,  $(\pi_{1}(1),\pi_{1}(2))=(2,1) $,  $ h_0=1$, $(h_{1,1}, h_{2,1},h_{3,1})=(1,2,2)$,  $(h_{1,2}, h_{2,2},h_{3,2})=(3,1,0)$ and $(h_{1,3}, h_{2,3},h_{3,3})=(2,1,3)$ in Theorem \ref{4}. Then $\left\lbrace \mathcal{F}^{0},\mathcal{F}^{1},\cdots,\mathcal{F}^{15}\right\rbrace $ forms a quaternary $(16,4,64,16)$-ZCCS, where  $ \mathcal{F}^{3} $  and $ \mathcal{F}^{10} $ are given by
	$$\begin{bmatrix}
		\mathbf{f}^{3}_0\\\mathbf{f}^3_1\\\mathbf{f}^3_2\\\mathbf{f}^3_3
	\end{bmatrix}=
	\begin{bmatrix}
		1     2     1     2     1     3     3     1     3     2     3     2     3     3     1     1     1     2     1     2     1     3     3     1     3     2     3     2     3     3     1     1     1     2     1     2     1     3     3     1     3     2     3     2     3     3     1     1     1     2     1     2     1     3     3     1     3     2     3     2     3     3     1     1
		\\
		1     2     1     2     2     0     0     2     1     0     1     0     2     2     0     0     1     2     1     2     2     0     0     2     1  0     1     0     2     2     0     0     1     2     1     2     2     0     0     2     1     0     1     0     2     2     0     0     1     2 1     2     2     0     0     2     1     0     1     0     2     2     0     0
		\\
		1     2     1     2     3     1     1     3     3     2     3     2     1     1     3     3     1     2     1     2     3     1     1     3     3     2     3     2     1     1     3     3     1     2     1     2     3     1     1     3     3     2     3     2     1     1     3     3     1     21     2     3     1     1     3     3     2     3     2     1     1     3     3
		\\
		1     2     1     2     0     2     2     0     1     0     1     0     0     0     2     2     1     2     1     2     0     2     2     0     1     0     1     0     0     0     2     2     1     2     1     2     0     2     2     0     1     0     1     0     0     0     2     2     1     21     2     0     2     2     0     1     0     1     0     0     0     2     2
		
	\end{bmatrix}$$
	$$\begin{bmatrix}
		\mathbf{f}^{10}_0\\\mathbf{f}^{10}_1\\\mathbf{f}^{10}_2\\\mathbf{f}^{10}_3
	\end{bmatrix}=
	\begin{bmatrix}
		1     1     3     3     1     2     1     2     3     1     1     3     3     2     3     2     3     3     1     1     3     0     3     0     1     3     3     1     1     0     1     0     1     1     3     3     1     2     1     2     3     1     1     3     3     2     3     2     3     3 1     1     3     0     3     0     1     3     3     1     1     0     1     0
		\\
		1     1     3     3     2     3     2     3     1     3     3     1     2     1     2     1     3     3     1     1     0     1     0     1     3     1     1     3     0     3     0     3     1     1     3     3     2     3     2     3     1     3     3     1     2     1     2     1     3     3 1     1     0     1     0     1     3     1     1     3     0     3     0     3
		\\
		1     1     3     3     3     0     3     0     3     1     1     3     1     0     1     0     3     3     1     1     1     2     1     2     1     3     3     1     3     2     3     2     1     1     3     3     3     0     3     0     3     1     1     3     1     0     1     0     3     3 1     1     1     2     1     2     1     3     3     1     3     2     3     2
		\\
		1     1     3     3     0     1     0     1     1     3     3     1     0     3     0     3     3     3     1     1     2     3     2     3     3     1     1     3     2     1     2     1     1     1     3     3     0     1     0     1     1     3     3     1     0     3     0     3     3     3  1     1     2     3     2     3     3     1     1     3     2     1     2     1
		
	\end{bmatrix}$$
	The sum of aperiodic auto-correlation of sequences $ \mathcal{F}^3 $ is presented in Figure  \ref{Figure1} and the sum of aperiodic cross-correlation of sequences $ \mathcal{F}^3  $ and $ \mathcal{F}^{10} $ is presented in Figure \ref{Figure2}.
	
\end{ex}

\begin{figure}[htbp]
	\begin{minipage}[t]{0.48\textwidth}
		\centering
		\includegraphics[width=1.1\textwidth,height=1.1\textwidth]{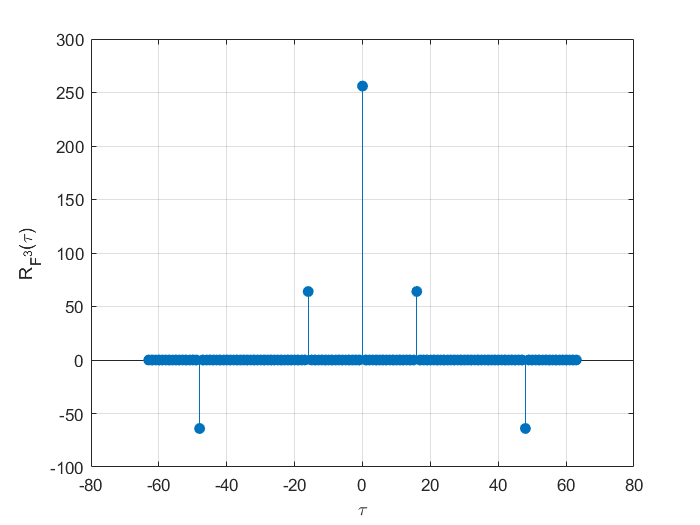}
		\caption{Auto-correlation of $ \mathcal{F}^3 $}
		\label{Figure1}
	\end{minipage}
	\begin{minipage}[t]{0.48\textwidth}
		\centering
		\includegraphics[width=1.1\textwidth,height=1.1\textwidth]{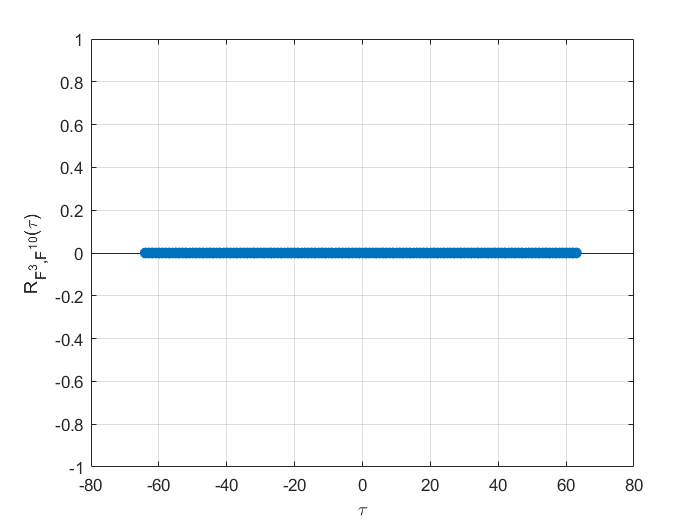}
		\caption{Cross-correlation of  $ F^3  $ and $ \mathcal{F}^{10} $}
		\label{Figure2}
	\end{minipage}
\end{figure}
\section{Comparison}
\textcolor{black}{\begin{table}
		\centering
		\setlength{\abovecaptionskip}{0cm}
		\setlength{\belowcaptionskip}{0.2cm}
		\caption{Summary of Existing MOCSs}
		\tiny
		\begin{tabular}{|m{1cm}<{\centering}|m{2cm}<{\centering}|m{5.8cm}<{\centering}|m{5.2cm}<{\centering}|}
			\hline Source&Based on&Parameters&Conditions\\
			\hline \cite{AR}&GBF&$(2^k,2^k,2^m)$&$0<k\leq m$\\
			\hline \cite{SW}&GBF&$(2^{k'},2^{k+1},2^m+2^t)$&$0<k,t\leq m;0\leq {k'}\leq t; {k'}\leq {k-1}$\\
			\hline \cite{SW}&GBF&$(2^k,2^{k+1},2^m+2^t)$&$0<k\leq t\leq m$\\
			\hline \cite{LT}&GBF&$(2^k,2^{k+1},2^m+2^t)$&$0\leq t<k\leq m$\\
			\hline\cite{KM}&GBF&$(2^{k+1},2^{k+1},2^{m-1}+2^{m-3})$&$k\leq m-5$\\
		\hline\cite{SP5}&$q$-ary function&$(p^{n+1},p^{n+1},p^{m})$&$p$  is a prime number, $0<n<m$\\
			\hline \cite{PS2}&MVF&$(\prod^{k}_{i=1}p_{i}^{n_i},\prod^{k}_{i=1}p_{i}^{n_i},\prod^{k}_{i=1}p_{i}^{m_{i}})$&$p_{i}|q,$ $q$ is a finite positive integer, $i\in \mathbb{N}_k$,$0<n_i\leq m_i$\\
			\hline \cite{SD1}&PU matrix&$(M,M,M^m)$&$m>0$, $M$ is the order of  PU matrix\\
			\hline \cite{SD2}&PU matrix&$(M,M,N^m)$&$N|M,m>0$, $M$ is the order of  PU matrix\\
			\hline \cite{YJ}&Kronecker product&$(M_{1}M_{2},M_{1}M_{2},N_{1}N_{2})$&$(M_1,M_1,N_1)$-CCC  and $ (M_2,M_2,N_2)$-CCC, $M1, M2,N_1,N_2$ are four even numbers\\
			\hline \cite{ZG}&Kronecker product&$(M,M,MN_1N_2)$&$(M,M,N_1)$-CCC  and $ (M,M,N_2)$-CCC, $2\leq M,N_1,N_2$\\
		\hline \cite{KL}&Extended correlation&$(MP,MP,2N-1)$&$(M,M,N)$-CCC and $  (P, P, N)$-
			CCC,  $2\leq M,P,N$\\
\hline \cite{SB}&concatenation&$(M_1M_2/2,M_1M_2/2, 2L_1 L_2)$&$M_1, M_2$ are two even numbers, $ L_1,L_2$ are two positive integers\\
			\hline Theorem \ref{thm1}&EBF&$(q^{d'},q^{v+d},a_mq^{m-1}+\sum^{v-1}_{k=1}a_{k}q^{m-v+k-1}+q^{u} )$&$0<d'<d< m$, $ a_{k}\in \mathbb{Z}_q$, $a_m\in \mathbb{Z}^*_q$, and $q\geq 2$ is a  positive integer\\
			\hline
		\end{tabular}
\end{table}}

\begin{table}
	\centering
	\setlength{\abovecaptionskip}{0cm}
	\setlength{\belowcaptionskip}{0.2cm}
	\caption{Summary of Existing ZCCSs}
	\tiny
	\begin{tabular}{|m{0.8cm}<{\centering}|m{1.5cm}<{\centering}|m{6cm}<{\centering}|m{3cm}<{\centering}|m{0.8cm}<{\centering}|m{1cm}<{\centering}|}
		\hline Source&Based on&Parameters&Conditions&Optimal&Remark\\
		\hline \cite{XC}&GBF&$(2^{k+1},2^{k+1},3\cdot2^{m},2^{m+1})$&$0<k\leq m$&$\times$& Direct\\
		\hline \cite{XC}&GBF&$(2^{k+2},2^{k+2},2^{m}\cdot L,2^{m}\cdot L')$&$L'>\frac{L}{2}$&$\surd$& Direct\\
		\hline \cite{XC}&GBF&$(2^{k+1},2^{k+1},3\times 2^{m},2^{m+1})$&$m>0,k>0$&$\surd$& Direct\\
		\hline \cite{WS}&GBF&$(2^{k+v},2^{k},2^{m},2^{m-v})$&$v\leq m,k\leq m-v$&$\surd$& Direct\\
		\hline
		\cite{SP3}&GBF&$(2^{n},2^{n},2^{m-1}+2,2^{m-2}+2^{\pi(m-3)}+1)$&$\pi$ is a permutation of $\mathbb{N}_{m-2}$, $m\geq 3$&$\surd$& Direct\\
		\hline \cite{SP3}&GBF&$(2^{n+1},2^{n+1},2^{m-1}+2,2^{m-2}+2^{\pi(m-3)+1})$&$\pi$ is a permutation of $\mathbb{N}_{m-2}$, $v\leq m$, $q\geq 2$, $m\geq 2$&$\surd$& Direct\\
		\hline \cite{SP2}&GBF&$(2^{n+p},2^{n},2^{m},2^{m-p})$&$p\leq m$&$\surd$& Direct\\
		\hline \cite{SP1}&GBF&$(2^{k+p+1},2^{k+1},2^{m},2^{m-p})$&$k+p\leq m$&$\surd$& Direct\\

		\hline \cite{GG1}&GBF&$(2^{k+1},2^{k+1},3(2^{m-1}+2^{m-3}),2(2^{m-1}+2^{m-3}))$&$m\geq 5$, $k>0$&$\surd$& Direct\\
		\hline \cite{GG1}&GBF&$(R2^{k+l},2^{k+1},R(2^{m-1}+2^{m-3}),2^{m-1}+2^{m-3})$&$m\geq 5$, $k>0$, and $R$ is even&$\surd$& Direct\\
		 \hline \cite{DS}&BH Matrix&$(MP,M,MP,M)$&$M,P$ are the order of  BH matrix &$\surd$& Indirect\\
		\hline \cite{DS}&Optimal ZPU Matrix& $(MP,M,M^{N+1}P,M^{N+1})$&$M,P$ are the order of  BH matrix, $N>0$&$\surd$& Indirect\\
		\hline \cite{AA}&Hadamard product &$(2^{n+1},2^{n+1},N,Z)$&$N\geq 3,$ $N$ is odd, $\lfloor \frac{N}{Z}=1\rfloor$&$\surd$& Direct\\
		\hline \cite{AA}&ZCP&$(2^{m},2^{m},L,Z)$&$Z\geq \lceil \frac{L}{2}\rceil$&$\surd$& Direct\\
	
		\hline \cite{GG}&PBF&$(\prod^{l}_{i=1}p_{i}2^{n+1},2^{n+1},2^{m}\prod^{l}_{i=1}p_{i},2^{m})$&$\forall p_{i}$ is a prime, $n,m>0$ &$\times$& Direct\\
		\hline \cite{SP4}&PBF&$(p2^{k+1},2^{k+1},p2^{m},2^{m})$&$p$ is a prime&$\surd$& Direct\\

		\hline \cite{RA}&MVF& $(\prod^{k}_{i=1}p_{i}^2,\prod^{k}_{i=1}p_{i},\prod^{k}_{i=1}p_{i}^{m_{i}},\prod^{k}_{i=1}p_{i}^{m_{i}-1})$&$p_i$ is a prime number, $m_i>0$&$\surd$& Direct\\
		\hline \cite{SB}&EBF&$(q^{v+1},q,q^{m},q^{m-v})$&$q\geq 2$, $v\leq m$&$\surd$& Direct\\
		\hline Theorem \ref{4}&EBF&$(q^{v+d},q^d,q^{m},q^{m-v})$&$v< m$, $d\leq m-v$, $q\geq 2$ is a  positive integer&$\surd$& Direct\\
		\hline
	\end{tabular}
\end{table}
Table 1 and Table 2 show the existence of constructions of MOCSs and ZCCSs in previous papers. The notation ``$\surd$" (resp. ``$\times$") in Table 2 means the corresponding ZCCSs are optimal (resp. non-optimal).

\textcolor{black}{
From Table 1, we know that all GBFs based MOCSs  have lengths of $2^m$ or $2^{m}+2^{t}$ \cite{ AR,  SW, LT,KM}. The constructions in \cite{SP5} and \cite{PS2} generate MOCSs with flexible lengths by using $q$-ary functions and MVFs, respectively. But  both of these methods only have power of two lengths when $q=2.$ Other methods for designing  MOCSs include  PU matrices \cite{SD1,SD2}, interleaving, Kronecker product \cite{ZG,YJ}, extended correlation \cite{KL} and concatenation \cite{SB}. However, These  methods are hard to be applied in engineering due to their large space and time requirements in hardware generation. Compablack with the previous constructions, our results have flexible lengths and  non-power-of-two lengths when $q=2$.
}

\textcolor{black}{
From Table 2, we know the constructions of ZCCSs in the literature mainly divided into direct and indirect approaches. The direct methods are mainly based on GBFs \cite{SP3,SP2,SP1,WS,XC, GG1}, Pseudo-Boolean functions (PBFs) \cite{GG,SP4}, EBFs \cite{SB} and MVFs \cite{RA}. In fact, all existing ZCCSs constructed based on GBFs and PBFs have multiples of two lengths and the ZCCSs based on MVFs have limited set sizes.  For other indirect methods, some researchers provided ZCCSs by Hadamard product \cite{AA}, Z-complementary pairs (ZCPs) \cite{AA},  BH matrix  and optimal Z-paraunitary (ZPU) matrices \cite{DS}. These constructs are difficult to implement on hardware.  In the case of the same length and ZCZ width, compablack to \cite{SB}, the proposed  ZCCSs have larger set sizes or lengths. Moreover, our ZCCSs can accommodate more users on the basis of achieving the optimality.}

\section{Conclusion}
\textcolor{black}{In this paper, we mainly present a construction of optimal ZCCSs and  a construction of MOCSs with flexible lengths based on EBFs.  According to the arbitrariness of $q$, the proposed MOCSs cover the result in \cite{PS2} and have non-power-of-two lengths when $q=2$.
Moreover, the resulting MOCSs can be obtained directly from EBFs without using  tedious sequence operations. The proposed MOCSs with flexible lengths  find many applications in wireless communication due to its good correlation properties.
The proposed ZCCSs are optimal with respect to the theoretical upper bound and  we can obtain a new class of ZCCSs of arbitrary lengths with large zero correlation zone width. }

\section*{Declarations}
\textbf{Funding}  This research was supported by the National Natural Science Foundation of China (Grant No.
12171241)\\
\textbf{Conflicts of Interest} The authors declare that they have no conflicts of interest.\\
\textbf{Ethics approval and consent to participate} Not applicable.\\
\textbf{Consent for publication} Not applicable.

\end{document}